\documentclass[a4paper,10pt]{article}
\usepackage[english]{babel} 
\usepackage{a4}						
\usepackage{tabularx} 				
\usepackage[onehalfspacing]{setspace}

\usepackage{amsmath}
\usepackage{amsfonts}
\usepackage{amssymb} 
\usepackage{amsthm}
\usepackage{pst-all}        
\usepackage{float}          
\usepackage{graphicx}       

\newcommand{\BfGOP}[1]{\mathop{\mathchoice%
{\raise-0.22em\hbox{\huge $#1$}}%
{\raise-0.05em\hbox{\Large $#1$}}{\hbox{\large $#1$}}{#1}}}
\newcommand{\bigtimes}{\BfGOP{\times}}

\newcommand{\R}{\mathbb{R}}

\newcommand{\E}{\mathbb{E}}

\newcommand{\HH}{\mathcal{H}}

\newcommand{\pre}{\text{pre}}
\newcommand{\suc}{\text{suc}}

\theoremstyle{plain}
\newtheorem{theorem}{Theorem}
\newtheorem{remark}[theorem]{Remark}
\newtheorem{definition}[theorem]{Definition}
\newtheorem{lemma}[theorem]{Lemma}
\newtheorem{example}[theorem]{Example}
\newtheorem{question}[theorem]{Question}

\begin{document}
\begin{flushright}
    ZMP-HH / 15-8 \\
    Hamburger Beitr{\"a}ge zur Mathematik Nr. 538 
  \end{flushright}
	\vspace{0.5cm}
	\begin{center}
	\Large
	Nash Equilibria And Partition Functions \\ Of Games With Many Dependent Players
	\end{center}
	\vspace{0.5cm}
  \begin{center}	
    Elisabeth Kraus,
		TQA Research Group, University of Munich.
  \end{center}
	\begin{center}	
    Simon Lentner\footnote{Corresponding author:
      \texttt{simon.lentner@uni-hamburg.de}},
    Algebra and Number Theory, University of Hamburg.
  \end{center}
	\vspace{0.5cm}

\begin{abstract}
	We discuss and solve a model for a game with many players, where a subset of truely 
	deciding players is embedded into a hierarchy of dependent agents. 
	
	These interdependencies modify the game matrix and the Nash equilibria for the 
	deciding players. In a concrete example, we recognize the partition function of the 
	Ising model and for high dependency we observe a phase transition to a new Nash 
	equilibrium, which is the Pareto-efficient outcome. 
	
	An example 
	we have in mind is the game theory for major shareholders in a stock market, where 
	intermediate companies decide according to a majority vote of their owners and compete 
	for the final profit. In our model, these interdependency eventually forces cooperation.
\end{abstract}

\tableofcontents

\section*{Acknowledgements}
Both authors would like to thank Prof. Martin Schottenloher and the research group TQA at the LMU Munich for hospitality and support. The second author is supported by the DFG RTG 1670, University of Hamburg.

\section{Introduction}
Roughly 20 years ago, an exciting new set of methods has been introduced into
game theory: An underlying game $G$, such as the Minority Game, played by a
large ensemble of players, can be analyzed and often solved using methods from
statistical physics. For comprehensive overviews on the development of this
subject, see \cite{CMZ} or \cite{Coolen}. The analysis exhibits critical points
where phase transitions appear in the thermodynamic limit of many players and it 
provides models for the emergence of mutual cooperation. The evolution of a large set of agents with prescribed strategies or learning mechanisms have been studied as dynamical systems e.g. in \cite{ChalletZhang}, \cite{Coolen} or \cite{NM}.\\ 

In this article, we want to focus on the influence of a large
ensemble of agents on the game theory of a given game $G$, including
game matrix and Nash equilibria. In the model we study, the game $G$ is
embedded into a hierarchy of automata/agents, that pass decisions by majority
votes. The top of the hierarchy remains a set of new active/deciding
players, which now play a new transformed game $\Gamma$. We then wish to understand how the game theory of $\Gamma$ compared to $G$ changes, depending on the given intermediate hierarchy. This has been solved by the first author as part of her diploma thesis \cite{DA}.\\

More specifically, in section 2 we suppose that we are given a
weighted directed graph $\HH$ and a game $G=\langle L, S,u \rangle$ played by a subset of the vertices
$\HH_V$ of $\HH$ called \emph{executive players} $L\subset \HH_V$. We define a
transformed game $\Gamma = \langle \Lambda, \Sigma, \nu \rangle$ with new
players $\Lambda\subset \HH_V$. The graph hereby is imagined as a hierarchy of
agents with executive players $L$ at the bottom of the hierarchy, while new
\emph{deciding players} $\Lambda$ at the top of the hierarchy successively control the
behaviour of the agents according to their influence. Conversely, the \emph{payoff of
$G$} for all executive players $L$ is finally collected by the deciding players
according to the natural bargaining process in this situation (Shapely value). 
Other payoff mechanisms are possible and got discussed in \cite{DA} as well.\\

The main example for this model we have in mind is the stock market, where the deciding players $\Lambda$ send instructions through a graph that represents the
structure of the mutual ownerships of the companies. Finally, some executive companies $L$ play a executive game $G$ in ``reality'', such as prisoner's dilemma or minority game. So we ask, how the stock market and mutual ownerships of the intermediate
agents/companies have altered the game $G$ to $\Gamma$. \\

To solve $\Gamma$ thermodynamically, we need more than a partition function
summing over all possible strategies as done e.g.\ in \cite{Coolen}. Rather, the
conditional probabilities of one agent's decision influencing another one have
to be calculated. They correspond physically to $k$-point correlation functions (see definition \ref{def_k-point}) and especially for $\Lambda=\emptyset$ (no deciding players) the overall expression reduces to the partition function.\\ 

In section 3, we recognize that for treelike hierarchy graphs our $k$-point correlation functions coincide with a generalized Ising model on the graph $\HH$. This enables us in principle to write down the game matrix, Nash equilibria and phase transitions of the game $\Gamma$ for a given game $G$ whenever the Ising model for $\HH$ is accessible. Of particular interest to us is the case where $\HH$ is a random graph, which has been solved in \cite{Dommers}. \\

In section 4 we demonstrate the approach and
methods developed in this article. We solve and thoroughly analyse an example of
an executive prisoners dilemma $G$ being transformed to a hierarchical game $\Gamma$
with again two deciding players, but with a certain hierarchy of agents between
deciders and executive $G$-players.\\

Especially, we can establish a phase transition in the game $\Gamma$, if the
branching factor of the tree is sufficiently high (otherwise we get only a tipping point)
as the mutual influence approaches a critical threshold. The phase
transition in $\Gamma$ separates a phase with the (defecting) Nash equilibrium in
$G$ from a phase corresponding to the (cooperating) Pareto-efficient outcome.
Roughly spoken, if the mutual dependency in decision making gets high, egoistic
strategies become unstable and mutual cooperation emerges.

\section{Definition Of The Hierarchical Game} \label{definition}

In the following we suppose to be given a weighted directed graph $\HH$, whose vertices 
contain among others executive players $L$ playing a game $G$. We suppose $G$ to have only two moves. The graph should be imagined as a hierarchy with executive players $L$ at the bottom of the hierarchy. The key notion of this article is then a transformed game $\Gamma = \langle \Lambda, \Sigma, \nu
\rangle$ with new deciding players $\Lambda$ at the top of the hierarchy, who
successively control the behaviour of the agents and collect the $G$-payoff
according to their influence. 

\begin{figure}[H]
\begin{center}
\begin{pspicture}(-3,0.5)(8,5)
\cnode*(1,1){2pt}{A}
\rput(0.5,1){\rnode{A'}{1}}
\cnode*(3,1){2pt}{E}
\rput(3.5,1){\rnode{E'}{2}}
\cnode*(5,1){2pt}{B}
\rput(5.5,1){\rnode{B'}{3}}
\rput(-1.2,1.2){executive}
\rput(-1.2,0.8){players L:}
\rput(-1.2,4.2){deciding}
\rput(-1.2,3.8){players \(\Lambda\):}
\cnode*(1,4){2pt}{C}
\rput(0.5,4){\rnode{C'}{\(\lambda_1\)}}
\cnode*(3,4){2pt}{F}
\rput(3.5,4){\rnode{F'}{\(\lambda_2\)}}
\cnode*(5,4){2pt}{D}
\rput(5.5,4){\rnode{D'}{\(\lambda_3\)}}
\psframe[framearc=0.2](-0.1,0)(6,1.4)
\rput(3,0.4){\rnode{S}{L playing the executive game $G$}}
\fontfamily{hlx}\selectfont 
\rput[bl](1.8,2.05){\pscharpath[fillstyle=solid,fillcolor=lightgray,
linewidth=1pt]{\fontsize{1cm}{1cm} \selectfont graph}}
\pnode(2,1.5){a}
\pnode(1.5,2){a'}
\ncline{->}{a}{A}
\ncline{->}{a'}{A}
\pnode(4,1.5){b}
\pnode(4.5,2){b'}
\ncline{->}{b}{B}
\ncline{->}{b'}{B}
\pnode(1.5,3){c}
\pnode(2,3.5){c'}
\ncline{<-}{c}{C}
\ncline{<-}{c'}{C}
\pnode(4.5,3){d}
\pnode(4,3.5){d'}
\ncline{<-}{d}{D}
\ncline{<-}{d'}{D}
\pnode(2.5,3.5){f}
\pnode(3.5,3.5){f'}
\ncline{<-}{f}{F}
\ncline{<-}{f'}{F}
\pnode(2.5,1.5){e}
\pnode(3.5,1.5){e'}
\ncline{->}{e}{E}
\ncline{->}{e'}{E}
\pnode(-2.8,4){V}
\pnode(-2.8,1){V'}
\ncline{->}{V}{V'}\ncput*{instructions for G-moves}
\pnode(7.3,4){W}
\pnode(7.3,1){W'}
\ncline{<-}{W}{W'}\ncput*{payoff floating back}
\end{pspicture}
\end{center} 
\end{figure}

\begin{definition}[Executive Game $G$]
  From now on, let $G=\langle L, \{S_i\}_{i\in L}, \{u_i\}_{i\in L} \rangle$ be
  a game in normal form with  players $L=\lbrace 1 \ldots n \rbrace$ and each
  player $i\in L$ having two strategies $S_i=S=\lbrace \pm 1 \rbrace$. The
  overall strategy set is hence 
  $S^L = \underset{i\in L}{\bigtimes} S_{i}=\lbrace  \pm 1 \rbrace ^n$ and we
  denote the payoff for each player $i\in L$ by $u_i:S^L\longrightarrow  \R$. 
\end{definition}

We denote by $A^B$ the set of all maps between $A$ and $B$ and by $\R A$ the vector space spanned by the set $A$.

\begin{definition}[Hierarchy Graph $\HH$]
  Let $\HH=\left(\HH_V,\HH_E,\{f_{vw}\}_{vw\in\HH_E}\right)$ be a
  connected, directed, weighted graph with vertex set $\HH_V$ and directed
  edges $vw\in\HH_E$ with positive weights $f_{vw}>0$ for $v,w\in \HH_V$.\\
  
  We denote the direct predecessors and successors of vertices $w,v\in\HH_V$ by
  $$\pre(w)=\lbrace  v \in \HH_V \ \vert \ vw \in \HH_E \rbrace
  \qquad \suc(v)=\lbrace  w \in \HH_V \ \vert \ vw \in \HH_E \rbrace.$$ 
  
  We further denote by $\HH_0\subset \HH_V$
  all vertices without predecessors and without loss
  of generality we assume the predecessor weights to be normed:
  $$\forall \ w\in \HH_V \setminus \HH_0: \; \sum_{v\in\pre(w)} f_{vw}=1$$
\end{definition}

\begin{definition}(Hierarchical Game $\Gamma=\HH G$)
Suppose a fixed game $G=\langle L, S, u \rangle$ and a fixed hierarchy graph
$\HH$ with $L\subset \HH_V$. The transformed hierarchical game $\Gamma = \HH
G:=\langle \Lambda, \Sigma, \nu \rangle$ consists of
\begin{itemize}
 \item A set of \emph{deciding players} $\Lambda:=\HH_0= \lbrace \lambda_1 \ldots
    \lambda_m \rbrace$.
 \item A \emph{strategy set} $\Sigma_\lambda=\Sigma=S^L$ for each deciding player
    $\lambda\in\Lambda$. Such a strategy formulates a
    $G$-strategy-command to each executive player $i\in L$. The
    overall strategy set is hence $\Sigma^\Lambda$.
 \item A \emph{payoff function} $\nu_\lambda:\Sigma^\Lambda\rightarrow \R$ for each
    deciding player $\lambda\in \Lambda$ given by 
    \begin{equation} \nu_\lambda=\left( \sum_{i\in L}\phi_\lambda^{(i)}
      \cdot u_i\right)\circ \pi\circ P_{L|\Lambda}. \end{equation}
    The function $\pi:\R\Sigma^L=\R(S^L)^L\rightarrow \R S^L$ is given by restricting a
    set of $G$-strategy-commands for each executive player
    $\left(\sigma_i^{(j)}\right)_{i,j\in L}\in \Sigma^L$ to the strategies
    chosen for the respective player $\left(\sigma_i^{(i)}\right)_{i\in
    L}\in S^L$.\\

    The functions $P_{L|\Lambda}:\R\Sigma^\Lambda\rightarrow \R\Sigma^L$ and
    $\phi_\lambda: \Lambda\rightarrow \R^L$ depending on the hierarchy
    $\HH$    will be defined in what follows:	
    \begin{itemize}
     \item $P_{B|A}:\R\Sigma^A\rightarrow \R\Sigma^B$ for subsets $A,B\subset
	\Sigma$ denotes the \emph{conditional influence}
	of players $A$ on players $B$ and should be read as a
	($|A|+|B|$)-point-function. A deciding player $\lambda\in\Lambda$ has been defined to have a strategy 
	$$\sigma_\lambda
	  =(\sigma_\lambda^{(i)})_{i\in L}
	  \in S^L=:\Sigma$$
        	formulating the aim to have each executive player $i$ using strategy
	$\sigma_{\lambda}^{(i)}$. These $G$-strategy-commands
	$\sigma=(\sigma_\lambda)_{\lambda\in\Lambda}\in \Sigma$ of all
	deciding players $\lambda$ compete along the hierarchy graph and
	determine an overall outcome probability distribution 
	$P_{L|\Lambda}(\sigma)\in \R\Sigma^L$ as described in the next section.
      \item $\phi_\lambda:\Lambda\rightarrow \R^L$ describes, how much of the
	payoff
	earned by each	of the exe\-cutive player $i\in L$ can be finally
	collected by a deciding player $\lambda$. The condition 
	$\sum_{\lambda\in\Lambda} \phi_\lambda^{(i)}=1$ is needed. In \cite{DA} we
	have discussed different payoff collection mechanisms, but in
	the following we will restrict ourselves to the natural result of a
	bargaining process between the	deciding players $\Lambda$ determined by
	the the Shapely value	(\cite{Shapley}). This particular choice has
	moreover the nice property to only depend on the conditional
	influences $P_{L|\Lambda}$.
    \end{itemize}
  \end{itemize}	
\end{definition}

\begin{remark}Stock Market \\
An easy application of this model is a stock market. The game $\Gamma = \langle
\Lambda, \Sigma, \nu \rangle$ is played by \emph{deciding players} $\Lambda$
(e.g.\ major stockholders). The graph represents mutual ow\-ner\-ships of companies
that pass the instructions of the deciding players via
(\ref{singlevote}) to the \emph{executive players}: These equations
represent a voting in each node, that weights the possessions of the direct
predecessors (respective direct owners) together with a small percentage $D$, the \emph{free float} of randomly voting minor stockholders. The
executive players $L$ play the game $G=\langle L, S, u \rangle$
according to the instructions they get - they act as agents and aren't players
in a game theoretical sense. The payoff which the executive players get is returned to
the deciding players weighted by the Shapely value: The more influence
deciding player $\lambda$ has on executive player $i$, the more is $\lambda$ getting of $i$'s payoff
$u_i$.  \\ \\ 
So we ask how the stock market and mutual ownerships of the intermediate
agents/ companies have altered the game $G$ to $\Gamma$.
Roughly we find that if the mutual dependency in decision making gets high, egoistic
strategies become unstable and mutual cooperation emerges. 
\end{remark}

\subsection{Conditional Influences}
We yet have to explain the function $P_{L|\Lambda}:\R\Sigma^\Lambda\rightarrow
\R\Sigma^L$. First consider a neighbourhood graph $\HH_{p,P}$ consisting of
a point $p$
with predecessors $P$. Suppose a {\it yes-no-}decision process, where
$\sigma_v^{(p)}\in \{\pm 1\}$ represents commands of each $v \in
P$ to $p$. The process shall be a vote in $p$, where every predecessor
$v$ has votes according to the weight $f_{vp}$ and a percentage of
$D\in \ ]0,1[$ votes randomly $\sim 
\mathcal{N}(0,\sigma_{\mathcal{N}}^2)$.

\begin{lemma}[Single Vote]
For the neighbourhood graph  $\HH_{p,P}$ the probability for a result
$+1$ in the point $p$ under some given condition
$(\sigma_v^{(p)})_{v \in P}$ is
$$P^{single}_{p|P}\left(\sigma_v^{(p)}=+1\ |\
(\sigma_v^{(p)})_{v \in P}\right) \ 
= 1-P_\mathcal{N}(0,\sigma_{\mathcal{N}}^2)(- C) \ \approx \
\frac{1}{2} - \frac{1}{2} \tanh(- aC)$$
\begin{equation}  \textnormal{at which } C=\frac{1-D}{D}\sum_{v \in P} f_{v p}\sigma_v^{(p)} \textnormal{ and  } a=\sqrt{\frac{2}{\pi \sigma_{\mathcal{N}}^2}}. \label{singlevote} \end{equation}
\end{lemma}
\begin{proof}
  Denote by $X^{(p)}\sim 
  \mathcal{N}(0,\sigma_{\mathcal{N}}^2)$ the Gaussian random variable of the
  random voters.
  \begin{align*}
   P^{single}_{p|P}\left(\sigma_v{(p)}=+1\ |\
    (\sigma_v^{(p)})_{v\in P}\right) & =P^{single}_{p|P}\left(DX^{(p)}+(1-D)\sum_{v \in P}
    f_{vp}\sigma_v^{(p)} \geq 0 \right) \\
  &=P^{single}_{p|P}\left(X^{(p)} \geq
    -\underbrace{\frac{1-D}{D}\sum_{v \in  P} 
    f_{vp}\sigma_v^{(p)}}_{C}\right) \\
  &= 1-P_\mathcal{N}(0,\sigma_{\mathcal{N}}^2)(- C)\\ 
  & \approx \ \frac{1}{2} - \frac{1}{2} \tanh(- aC)
  \textnormal{ with } a=\sqrt{\frac{2}{\pi \sigma_{\mathcal{N}}^2}}
  \end{align*}
  For approximation we use the similarity of the normal distribution and
  the tangens hyperbolicus (see \cite{DA}).
\end{proof}

Provided that $\HH$ does not contain directed cycles, the entire voting process goes on iteratively and we obtain straight-forward by induction for any $A\subset \HH_0,B\subset \HH$:

$$P_{B|A}(\sigma)=\sum_{\tau\in \{\pm1\}^{\HH_{V}},\;\tau|_A=\sigma}
    \tau|_B \cdot \prod_{p\in \HH_{V} \setminus \HH_0}P^{single}_{p|\pre(p)}
  \left(\tau|_p\;\bigg \vert \;\tau|_{\pre(p)}\right)$$

If $\HH$ does contain directed cycles then there is no terminating voting process. We nevertheless propose in complete analogy to statistical mechanics to assign in such a situation the conditional probabilities, which clearly reduce to the previous expression when no directed cycle is present:

  $$ P_{B|A}(\sigma)=\frac{1}{Z_{B|A}}
    \sum_{\tau\in \{\pm1\}^{\HH_{V}},\;\tau|_A=\sigma}
    \tau|_B \cdot \prod_{p\in \HH_{V} \setminus \HH_0}P^{single}_{p|\pre(p)}
  \left(\tau|_p\;\bigg \vert \;\tau|_{\pre(p)}\right)
	$$
	with the following now nontrivial normalization constant called partition function
  $$Z_{B|A}(\sigma):=
    \sum_{\tau\in \{\pm1\}^{\HH_{V}},\;\tau|_A=\sigma\;\;}
    \prod_{p\in \HH_{V} \setminus \HH_0}P^{single}_{p|\pre(p)}
  \left(\tau|_p\;\bigg \vert \;\tau|_{\pre(p)}\right).$$

This expression can be justified by a random experiment as follows:  
	Let the probability space be $\Omega:=\{\pm1\}^{\HH_{V}}$ with product measure 
	$$P(\tau):= \prod_{p\in \HH_{V} \setminus \HH_0}P^{single}_{p|\pre(p)}\left(\tau|_p\;\bigg \vert \;\tau|_{\pre(p)}\right).$$
	Take as events $\Omega_A(\sigma)\subset \Omega$ to be the event that holds $\tau|_A=\sigma$ and analogously for 			
	$\Omega_B(\sigma')$. Then the conditional probability for $\Omega_B(\sigma')$ under the condition $\Omega_A(\sigma)$ is 		
	defined as 		
	$$P_{B|A}=\frac{P\left(\Omega_B(\sigma')\cap \Omega_A(\sigma)\right)}{P\left(\Omega_A(\sigma)\right)}.$$
	Plugging in the product measure $P(\tau)$ and taking a formal linear combination over the outcome $\sigma'=\tau|_B$ yields 	
	the formula above.

\subsection{Payoff Mechanisms}
Once the conditional probabilities $P_{L|\Lambda}(\sigma)\in \R\Sigma^L$ for given
strategies $\sigma_\lambda$ of each deciding player $\lambda\in\Lambda$ have
been evaluated, this determines the behaviour of the executive players to 
$$\tau=\left(\pi\circ P_{L|\Lambda}\right)(\sigma)\in \R S^L.$$
This strategy produces in the game $G$ a payoff $u_i(\tau)$ for each executive
player $i\in L$. So how is this payoff collected finally by the executive players in $\Lambda$? 
Denote by $\phi_\lambda^{(i)}$ the percentage of payoff of executive player $i$ that is collected by deciding player $\lambda$, the \emph{payoff collecting mechanism}. Then the overall payoff function of the game $\Gamma$ is 
\begin{eqnarray}
 \nu_{\lambda} : \Sigma^{\Lambda} & \longrightarrow & {\R} \nonumber \\ 
 \nu_{\lambda} (\sigma) &=&  \left( \sum_{i \in L} \phi_{\lambda}^{(i)} \cdot u_i \right) \circ \pi \circ P_{L|\Lambda} (\sigma). \label{payoffmatrix}
\end{eqnarray}

\begin{example}[Payoff by shares]
The most intuitive payoff collecting mechanism for treelike hierarchy graphs is payoff proportional to the amount of shares of the executive player $i$ indirectly held by deciding player $\lambda$: 
\begin{align}
\phi_{\lambda}^{(i)} & =  \sum_{\textnormal{paths from } \lambda \textnormal { to }i} \
\ \prod_{vw \ \in \textnormal{ path}} f_{vw} \nonumber 
\end{align}
This fulfills $\sum_{\lambda \in \Lambda} \phi_{\lambda}^{(i)} = 1 \ \forall i$ (see \cite{DA}). An example for this payoff collecting me\-cha\-nism is the paying of dividends proportional to the amount of stocks held by an owner. 
\end{example}

However, in the following we will restrict ourselves to the natural result of a bar\-gai\-ning process between the deciding players $\Lambda$ according to their influence - the payoff
is hence determined by the Shapely value \cite{Shapley} of the obvious
coalition function: 
\begin{align*}
z_i:2^{\vert \Lambda \vert}& \longrightarrow {\R} \\ 
K & \longmapsto \frac{P(s_i=1|\sigma_\lambda^{(i)}=1 \ \forall \lambda \in K,
\sigma_\lambda^{(i)} =-1 \ \forall \lambda {\not\in} K) -
P(s_i=1|r_\lambda^{(i)} = -1 \ \forall \lambda)}{2P(s_i=1|r_\lambda^{(i)} = 1 \
\forall \lambda)-1} 
\end{align*}
This particular choice for a payoff collection mechanism has moreover the nice
pro\-perty to only depend on the conditional influences $P_{L|\Lambda}$.

\begin{lemma}
For the payoff function by Shapely value we get 
\begin{equation} \label{Shapley} 
\phi_\lambda^{(i)}=\sum_{K \subseteq \Lambda, \lambda \in K} \frac{(|K|-1)!(|\Lambda|-|K|)!}{|\Lambda|} \left(z_i(K)-z_i(K \setminus \{\lambda\})\right).
\end{equation}
\end{lemma}
\begin{proof}
The necessary scaling condition $\sum_{\lambda \in \Lambda} \phi_{\lambda}^{(i)} = 1$ is fulfilled (see \cite{Shapley}). 
\end{proof}

\section{Hierarchical Games On Trees Are Ising Models}
In the following section we prove that the conditional influence can be calculated by using an isomorphic Ising model. For the common definition of the Ising model, see for example \cite{Ising}. Some
first analogies are obvious: \\ \\
\begin{tabular}{c|c}
Ising model & hierarchical game\\ 
\hline \hline 
particles in a graph $\HH$             & same graph $\HH$\\ \hline
spin of particle $v$:                     & strategy of $v \in \HH_V$:    \\ 
$\sigma_v \in \{\pm 1 \}$                 & $\sigma_v=(\sigma_v^{(1)} \ldots
\sigma_v^{(n)}) \in \lbrace              \pm 1 \rbrace ^n $ \\ \hline 
interaction $J_{vw}$                      & (with D modified) weights $f_{vw}
\cdot \frac{1-D}{D}$  \\ \hline 
external magnetic field (is set to $0$)   & some systematic bias (not treated) 
\end{tabular} \vspace{12pt} \\ 
In addition to that, $k$-point-functions are needed for conditional probabilities. 

\begin{definition} \label{def_k-point}
Let $A, B$ be disjoint subsets of the nodes in the Ising model. Let the
nodes beyond $A$, named $N(A)$, be the nodes of $\HH_V \setminus (A \cup B)$,
that fulfill the following condition: Every path to any $v' \in B$ hits
at least one $v \in A$. With nodes beyond $B$ (named $N(B)$) defined
analogously, the \emph{nodes between $A$ and $B$} are $N(A, B):=\HH_V
\setminus (N(A) \cup N(B))$. \\
With $k=\vert A \vert + \vert B \vert$ the \emph{$k$-point-function} is 
\begin{align}
\lbrace \pm 1 \rbrace ^{\vert B \vert} \times \lbrace \pm 1 \rbrace ^{\vert
A \vert}  \longrightarrow & \ \R \nonumber \\ 
(\sigma', \sigma)  \longmapsto & \ \langle \sigma' \ \vert \ N(A, B) \ \vert \ \sigma \rangle \nonumber \\ 
& = \sum_{\sigma', \sigma \textnormal{ fixed}} \exp(-\beta H_{N(A, B)})
\end{align}
where $H_{N(A, B)}$ is the Hamiltonian function of the restricted graph
just containing the nodes $N(A, B)$ and $\beta=\frac{1}{k_B T}$ is the
inverse temperature in the Ising model. 
\end{definition}

\begin{lemma} 
Let $\sigma_v \in \{ \pm 1 \}$ be the spin of particle $v$, $\sigma \in \{ \pm 1
\}^{\vert A \vert}$ the spins of particles in $A$ and let the external
magnetic field be $0$. Then the conditional probability for $\sigma_v$ given
$\sigma$ is 
\begin{equation}
P(\sigma_v =1 \ \vert \ \sigma) = \frac{\langle \sigma_v =1 \ \vert \ N_{A, \{v\}} \
\vert \ \sigma \rangle}{ \sum_{\sigma_v} \langle \sigma_v \ \vert \ N_{A, \{v\}}
\ \vert \ \sigma \rangle} \label{probability_k-point}
\end{equation}
\end{lemma}
\begin{proof}
See \cite{DA}.
\end{proof}

\begin{remark}
The Ising model without external magnetic field and with constant interaction J
is exactly solvable in one dimension, see \cite{Ising}. 
In this case, the conditional probability for two particles $v$, $v'$ with
spins $\sigma$, $\sigma'$ and distance $a$ is 
\begin{equation}
P(\sigma' \ \vert \ \sigma) = \frac{\cosh^a(\beta J) \pm \sinh^a(\beta J)}{2
\cosh^a(\beta J)} \label{onedimension}
\end{equation}
where the case "$+$" occurs when $\sigma = \sigma'$ and "$-$" if $\sigma
\neq \sigma' $. 
This remark will be needed for the example in subsection \ref{example_step1}.
\end{remark}

\begin{theorem}
Every hierarchical game on a graph (as defined in section \ref{definition})
that fulfills the condition, that $\forall i \in L$ the restricted graph of the
nodes $N_{\Lambda, {i}}$ is a tree, is isomorphic to an Ising model such that
a process in the game that leads from fixed $\sigma_\lambda^{(i)} \ \forall
\lambda \in \Lambda$ to the strategy $\sigma_i^{(i)}=s_i$ of one $i \in L$ is
equivalent to a process in an Ising model on the same graph with interactions
$J_{vw}=f_{vw}\frac{1-D}{D}$, inverse temperature $\beta=\sqrt{\frac{2}{\pi 
\sigma_{\mathcal{N}}^2}}$ and no external magnetic field. 
\end{theorem}
Isomorphic hereby means that the conditional influence and the $k$-point-functions coincide. 
\begin{proof}
For details, see section 4 in \cite{DA}. \\ 
As the graph had to be specified to a tree-like graph, the process fixates the
strategies step by step. Therefore is it enough to look at the local fixation of
a strategy in one node and to compare (\ref{singlevote}) and
(\ref{probability_k-point}). Doing this gives the interactions and the inverse
temperature as mentioned above. 
\end{proof}

\begin{remark}
The restriction to tree-like graphs seems harsh on the first sight. However even a one-dimensional model shows interesting behaviour (see subsection \ref{example_step1}) and also the calculation of Ising models on random graphs as done in \cite{Dommers} typically requires the graph to be at least locally treelike. 
\end{remark}

\section{Solution Of The Hierarchical Game}

\subsection{Steps For The Solution}
Let $\Gamma = \langle \Lambda, \Sigma, \nu \rangle$ be a hierarchical game as
defined in section \ref{definition} with $\Lambda=\{\lambda_1, \lambda_2\}$ and
$L=\{1,2\}$. (The restriction to two players allows the use of payoff matrices.)

\begin{enumerate}

\item \textbf{Determining the conditional influence} \\
With two players in each $\Lambda$ and $L$, the conditional influence
depends on only four variables obtained from the Ising model and
equation (\ref{probability_k-point}).
\begin{align}
\textnormal{conditional influence on 1}: x& = P_{L|\Lambda}(s_1=1 \ \vert \
\sigma_{\lambda_1}^{(1)}=-1, \sigma_{\lambda_2}^{(1)}=1) \nonumber \\
y & = P_{L|\Lambda}(s_1=1 \ \vert \ \sigma_{\lambda_1}^{(1)}=1, \sigma_{\lambda_2}^{(1)}=1)
\nonumber \\
\textnormal{conditional influence on 2}: \overline{x}& = P_{L|\Lambda}(s_2=1 \ \vert \
\sigma_{\lambda_1}^{(2)}=1, \sigma_{\lambda_2}^{(2)}=-1) \nonumber \\
\overline{y} & = P_{L|\Lambda}(s_2=1 \ \vert \ \sigma_{\lambda_1}^{(2)}=1,
\sigma_{\lambda_2}^{(2)}=1) \nonumber 
\end{align}

\item \textbf{Building up the pre-payoff matrix}\\
The pre-payoff matrix assigns to every combination of strategies $\sigma=(\sigma_{\lambda_1},\sigma_{\lambda_2})$ in $\Sigma^\Lambda$ the expected payoff $u_i$ to executive player $i \in L$ in game $G$:
\begin{align}
\E [(u_1,u_2) \ \vert \ \sigma] & = (u_1,u_2) \circ \pi \circ P_{L | \Lambda} (\sigma) \nonumber \\
& = \sum_{s \in
S} P_{L|\Lambda}(s \ \vert \ \sigma_{\lambda_1}, \sigma_{\lambda_2}) \cdot u(s) \nonumber
\\ 
& = \sum_{s_1} \sum_{s_2} P_{L|\Lambda}(s_1 \ \vert \ \sigma_{\lambda_1}^{(1)},
\sigma_{\lambda_2}^{(1)}) \cdot P_{L|\Lambda}(s_2 \ \vert \ \sigma_{\lambda_1}^{(2)},
\sigma_{\lambda_2}^{(2)})\cdot u(s_1,s_2) \nonumber 
\end{align}

\item \textbf{Building up the payoff matrix}\\
To get the payoff matrix $\nu_\lambda$ for $\Gamma$ as defined in equation (\ref{payoffmatrix}) multiply every item in the pre-payoff matrix with 
$$\left( \begin{array}{cc} \phi_{\lambda_1}^{(1)} & \phi_{\lambda_2}^{(1)} \\ \phi_{\lambda_1}^{(2)} &
\phi_{\lambda_2}^{(2)} \end{array} \right)=\left( \begin{array}{cc} \frac{y-x}{2y-1} &
\frac{\overline{x}+\overline{y}-1}{2\overline{y}-1} \\ \frac{x+y-1}{2y-1} &
\frac{\overline{y}-\overline{x}}{2\overline{y}-1} \end{array} \right).$$

\item \textbf{Usual methods}\\
Now that there is a payoff matrix for the game, the usual game theoretical
methods can be applied to find Nash equilibria and phase
transitions.

\end{enumerate}

\subsection{Example To Step 1: Easy Hierarchical Graph} \label{example_step1}

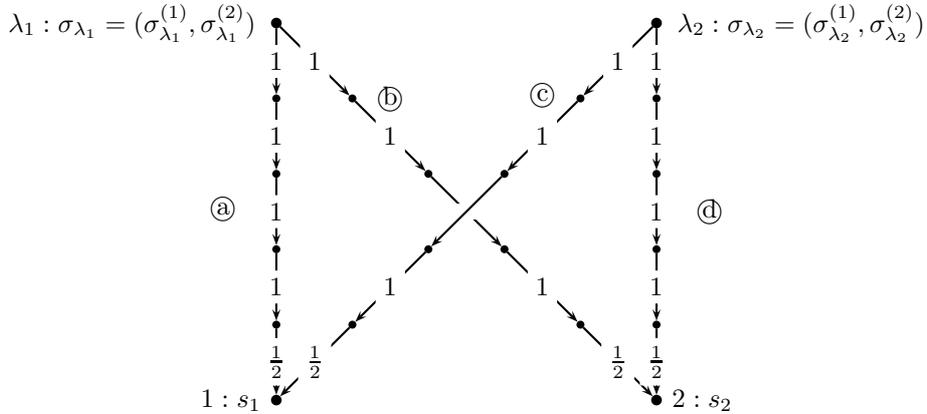
\begin{figure}[H]
\begin{center}
\begin{pspicture}(-2,1.5)(12.5,6)%
\psset{xunit=1cm, yunit=1cm, runit=1cm} 
\cnode*(2,1){2pt}{C}
\cnode*(7,1){2pt}{D}
\cnode*(2,6){2pt}{A}
\cnode*(7,6){2pt}{B}
\rput(0.1,6){\rnode{A'}{\(\lambda_1:\sigma_{\lambda_1}=(\sigma_{\lambda_1}^{(1)}
,\sigma_{\lambda_1}^{(2)})\)}}
\rput(1.4,1){\rnode{C'}{\(1:s_1\)}}
\rput(7.6,1){\rnode{D'}{\(2:s_2\)}}
\rput(8.9,6){\rnode{B'}{\(\lambda_2:\sigma_{\lambda_2}=(\sigma_{\lambda_2}^{(1)}
,\sigma_{\lambda_2}^{(2)})\)}}

\cnode*(2,5){1.5pt}{V}
\cnode*(3,5){1.5pt}{V'}
\cnode*(6,5){1.5pt}{V''}
\cnode*(7,5){1.5pt}{V'''}

\cnode*(2,4){1.5pt}{W}
\cnode*(4,4){1.5pt}{W'}
\cnode*(5,4){1.5pt}{W''}
\cnode*(7,4){1.5pt}{W'''}

\cnode*(2,3){1.5pt}{X}
\cnode*(5,3){1.5pt}{X'}
\cnode*(4,3){1.5pt}{X''}
\cnode*(7,3){1.5pt}{X'''}

\cnode*(2,2){1.5pt}{Y}
\cnode*(6,2){1.5pt}{Y'}
\cnode*(3,2){1.5pt}{Y''}
\cnode*(7,2){1.5pt}{Y'''}

\ncline{->}{A}{V}\ncput*{1}
\ncline{->}{V}{W}\ncput*{1}
\ncline{->}{W}{X}\ncput*{1}
\ncline{->}{X}{Y}\ncput*{1}
\ncline{->}{Y}{C}\ncput*{$\frac{1}{2}$}

\ncline{->}{A}{V'}\ncput*{1}
\ncline{->}{V'}{W'}\ncput*{1}
\ncline{->}{W'}{X'}\ncput*{ }
\ncline{->}{X'}{Y'}\ncput*{1}
\ncline{->}{Y'}{D}\ncput*{$\frac{1}{2}$}

\ncline{->}{B}{V''}\ncput*{1}
\ncline{->}{V''}{W''}\ncput*{1}
\ncline{->}{W''}{X''}
\ncline{->}{X''}{Y''}\ncput*{1}
\ncline{->}{Y''}{C}\ncput*{$\frac{1}{2}$}

\ncline{->}{B}{V'''}\ncput*{1}
\ncline{->}{V'''}{W'''}\ncput*{1}
\ncline{->}{W'''}{X'''}\ncput*{1}
\ncline{->}{X'''}{Y'''}\ncput*{1}
\ncline{->}{Y'''}{D}\ncput*{$\frac{1}{2}$}

\rput(1.3,3.5){\rnode{a}{\textcircled{a}}}
\rput(7.7,3.5){\rnode{c}{\textcircled{d}}}
\rput(3.5,5){\rnode{b}{\textcircled{b}}}
\rput(5.5,5){\rnode{d}{\textcircled{c}}}
\end{pspicture}
\end{center}
\caption{One-dimensional hierarchical game}
\end{figure}
In the one-dimensional hierarchical game, the deciding and executive players are
connected by chains of  $a,b,c$ or $d$ edges. Therefore the weights are $1$,
except down at the executive players (weights $\frac{1}{2}$). Let $D$ be
$\frac{1}{2}$. 
The conditional influence can now be calculated with equation
(\ref{onedimension}) which leads to 
\begin{align}
x&= P_{L|\Lambda}(s_1=1 \ | \ \sigma_{\lambda_1}^{(1)}=-1, \sigma_{\lambda_2}^{(1)}=1)=
\frac{\mu(-,a) \cdot \mu(+,c)}{\mu(-,a) \cdot \mu(+,c)+\mu(+,a) \cdot \mu(-,c)}
\nonumber \\ 
y&= P_{L|\Lambda}(s_1=1 \ | \ \sigma_{\lambda_1}^{(1)}=1, \sigma_{\lambda_2}^{(1)}=1)=
\frac{\mu(+,a) \cdot \mu(+,c)}{\mu(+,a) \cdot \mu(+,c)+\mu(-,a) \cdot \mu(-,c)}
\nonumber 
\end{align}
where
$\mu(\pm,k)=\cosh^{k-1}(\beta)\cosh(\frac{\beta}{2})\pm\sinh^{k-1}
(\beta)\sinh(\frac{\beta}{2})$. \\

\begin{figure}[H]

\includegraphics[width=6.5cm]{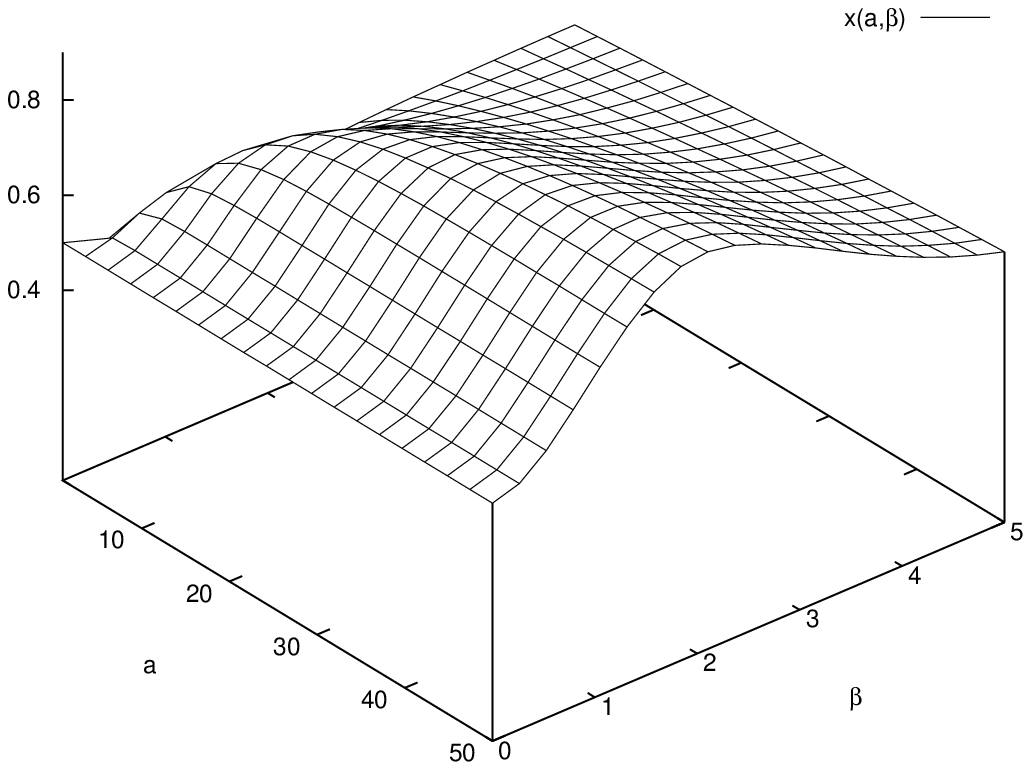}	   
\includegraphics[width=6.5cm]{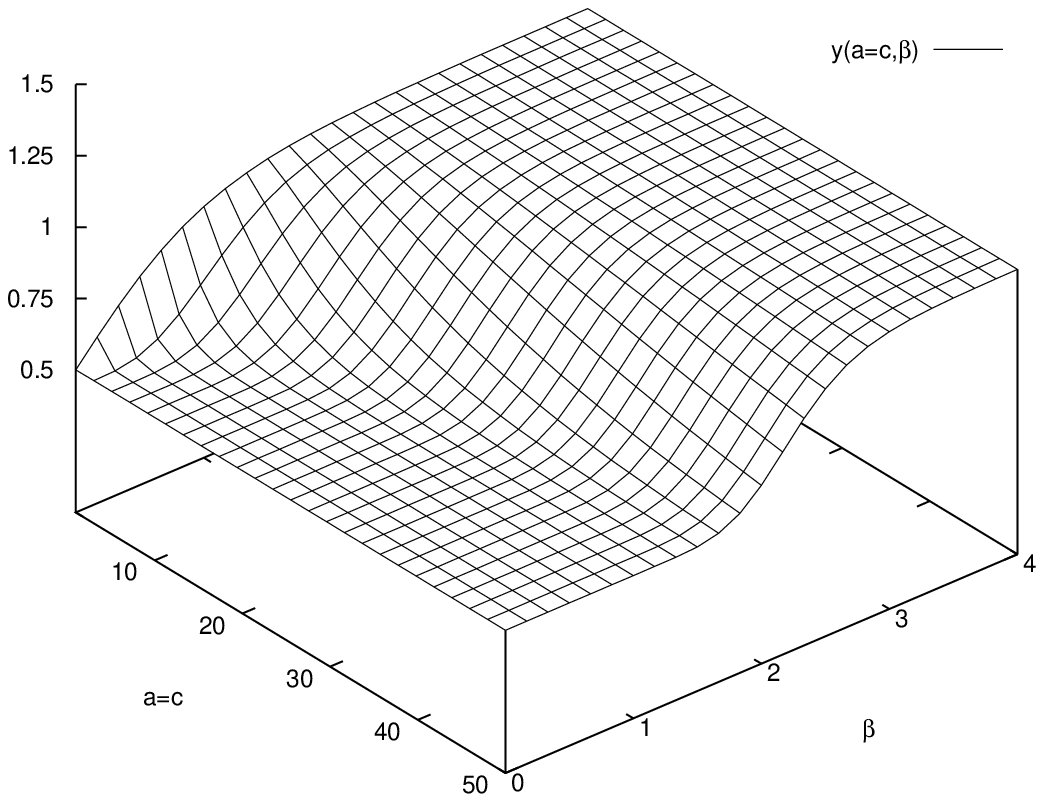}    

\caption{Conditional influence $x$ with fixed $c$ (left) and conditional influence $y$ with $a=c$ (right)}
\end{figure}

As $x$ is a measure for player $1$ obeying rather $\lambda_2$ than $\lambda_1$ if their
instructions differ, $x$ is growing if $a$ increases and the influence of
$\lambda_1$ therefore decreases as it can be seen on the left. However, for small $\beta$ (which means high
temperature $T$) $x$ is close to $\frac{1}{2}$ no matter how far the deciding
players are from each other. As $y$ shows how much player $1$ is likely to obey $\lambda_1$ and $\lambda_2$ if they
agree, the right graph shows how $y$ is close to $1$ if both deciding players are near to $1$ at a low temperature. If the distance and the temperature increase, $1$ tends to
choose its strategy randomly with probability $\frac{1}{2}$.

\subsection{Example To Step 2-4: Prisoner's Dilemma} \label{example_step2-4}
Let the game G be the well-known prisoner's dilemma with payoff matrix
\begin{equation}
\begin{array}{c||c|c} & C{\textnormal{\tiny{(ooperation)}}} & D{\textnormal{\tiny{(efection)}}} \\ \hline \hline  C & (1,1) & (-3,3) \\ \hline D &
(3,-3) & (-1,-1) \end{array} \ \ . \nonumber 
\end{equation}
Let the conditional influence be symmetric, so it goes down to just two
variables $x$ and $y$:
\begin{align}
x& = P_{L|\Lambda}(s_1=C \ \vert \ \sigma_{\lambda_1}^{(1)}=D,
\sigma_{\lambda_2}^{(1)}=C)=\overline{x} \nonumber \\
y& = P_{L|\Lambda}(s_1=C \ \vert \ \sigma_{\lambda_1}^{(1)}=C,
\sigma_{\lambda_2}^{(1)}=C)=\overline{y} \nonumber 
\end{align}
The complete payoff matrix has been calculated in \cite{DA}. \\
Depending on $x$ and $y$ the hierarchical game $\Gamma$ is isomorphic to one of the following games with unique Nash equilibrium $\hat \sigma$ and the tipping points for these three states are $x=\frac{2-y}{3}$ and $x=\frac{y+1}{3}$:

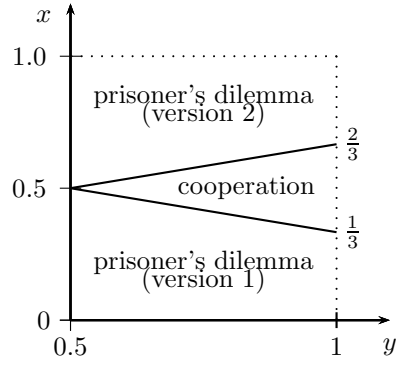
\begin{figure}[H]
\begin{pspicture}(-1,-0.3)(12,4)
\psset{xunit=7cm, yunit=3.5cm, Dy=0.5}
\rput(1.1,-0.1){\rnode{y}{\(y\)}}
\rput(0.45,1.15){\rnode{x}{\(x\)}}
\rput(0.75,0.85){\rnode{a}{prisoner's dilemma}}
\rput(0.75,0.78){\rnode{a'}{(version 2)}}
\rput(0.83,0.5){\rnode{b}{cooperation}}
\rput(0.75,0.22){\rnode{c}{prisoner's dilemma}}
\rput(0.75,0.15){\rnode{c'}{(version 1)}}
\rput(1.03,0.333333){\rnode{d}{\(\frac{1}{3}\)}}
\rput(1.03,0.666666){\rnode{d'}{\(\frac{2}{3}\)}}
\rput(0.5,-0.1){\rnode{e}{\(0.5\)}}
\rput(0.45,0){\rnode{e''}{\(0\)}}
\rput(1,-0.1){\rnode{e'}{\(1\)}}
\psline(1,0.03)(1,-0.03)
\psaxes[showorigin=false]{->}(0.5,0)(1.1,1.2)
\psline(0.5,0.5)(1,0.3333333)
\psline(0.5,0.5)(1,0.6666666)
\psline[linestyle=dotted](0.5,1)(1,1)
\psline[linestyle=dotted](1,0)(1,1)
\end{pspicture}
\caption{Illustration of the situation of $\Gamma$ depending on $x$ and $y$}
\end{figure}

\begin{enumerate}
\item $\Gamma$ is a \textbf{prisoner's dilemma} where   
\begin{align*}
\text{ for }\lambda_1:\ &(C,C)\simeq \textnormal{Cooperation}, \ (D,C)\simeq
\textnormal{Defection} \\
\text{ for }\lambda_2:\ &(C,C)\simeq \textnormal{Cooperation}, \ (C,D)\simeq
\textnormal{Defection} 
\end{align*}
That means, $\lambda_1$ identifies with $1$ and $\lambda_2$ identifies with $2$. \\
Hence the unique Nash equilibrium is $\hat \sigma = (\hat \sigma_{\lambda_1}, \hat \sigma_{\lambda_2})=((D,C),(C,D))$.

\item $\Gamma$ is a \textbf{prisoner's dilemma} where 
\begin{align*}
\text{ for }\lambda_1:\ &(C,C)\simeq \textnormal{Cooperation}, \ (C,D)\simeq
\textnormal{Defection} \\
\text{ for }\lambda_2:\ &(C,C)\simeq \textnormal{Cooperation}, \ (D,C)\simeq
\textnormal{Defection}
\end{align*}
That means, $\lambda_1$ identifies with $2$ and $\lambda_2$ identifies with $1$.\\
Hence the unique Nash equilibrium is $\hat \sigma =((C,D),(D,C))$.
\item \textbf{Cooperation}: \\ 
For $\lambda_1$ and $\lambda_2$, the strategy $(C,C)$ dominates every other
strategy and therefore $\hat \sigma=((C,C),(C,C))$.
\end{enumerate}

Because of the symmetry the payoffs of $\lambda_1$ and $\lambda_2$ coincide: $\nu_{\lambda_1}^{x,y}(\hat \sigma)=\nu_{\lambda_2}^{x,y}(\hat \sigma)$.
Hence the value of the game, i.e.\ the payoff in the Nash equilibrium $\hat \sigma$, is as follows: 
\begin{align}
\nu_{\lambda_1}^{x,y}(\hat \sigma)&=\left\lbrace \begin{array}{cl} -1+2x & x<\frac{2-y}{3} \\ -1+2y &
\frac{2-y}{3}<x<\frac{y+1}{3} \\ 1-2x & x>\frac{y+1}{3} \end{array} \right. 
\nonumber 
\end{align}

\begin{remark}
  In the 1-dimensional hierarchy considered in section \ref{example_step1} we had
  \begin{align}
  x&= 
  \frac{\mu(-,a) \cdot \mu(+,c)}{\mu(-,a) \cdot \mu(+,c)+\mu(+,a) \cdot
  \mu(-,c)}
  \nonumber \\ 
  y&= 
  \frac{\mu(+,a) \cdot \mu(+,c)}{\mu(+,a) \cdot \mu(+,c)+\mu(-,a) \cdot
  \mu(-,c)}
  \nonumber 
  \end{align}
  where $\mu(\pm,k)=\cosh^{k-1}(\beta)\cosh(\frac{\beta}{2})\pm\sinh^{k-1}
  (\beta)\sinh(\frac{\beta}{2})$ and the inverse temperature
  $\beta=\sqrt{\frac{2}{\pi\sigma_{\mathcal{N}}^2}}$ depended on the
  random minority voters.\\

  Hence in the one-dimensional hierarchy we get the tipping points above, but
  $x$ and $y$ are still smooth functions in $\beta$. On the contrary, for a
  two-dimensional hierarchy, $x$ and $y$ would exhibit proper non-analytical
  phase transitions in the thermodynamic limit, turning the tipping points into
  proper phase transitions.
\end{remark}

\section{Open Questions}

\begin{question}
	If we choose the Shapely value as payoff mechanism as above, the overall transformed 
	game $\Gamma$ depends only on the game $G$ and the correlators. What can be said in 
	general about the game theory of $\Gamma$ compared to $G$ \emph{without explicit 
	knowledge} of the correlators (under some reasonable, general assumptions)? 
\end{question}
\begin{question}
	It would be interesting to derive closed expressions for the correlators of an Ising 
	model on a locally treelike random graph, similarly to the partition functions obtained 
	in this case in \cite{Dommers}; it is to be expected that e.g. the $2$-point correlator 
	depends only on the distance. This would yield a very nice \emph{explicitly solvable 
	model} with phase transition for games on randomly dependent agents.
\end{question}
\begin{question}
	Our model does not necessarily require the graph to be a directed tree, see end of 
	section 2.1. In fact, mutual dependencies might be more realistic. Then the following 
	issues arise:
	\begin{itemize}
	\item Even in the easiest case, the partition sum does (to our surprise) not coincide 
	with the partition sum of the Ising model. Rather, there are corrections for every 
	directed loop. It would be nice to explain this behaviour and/or derive expressions for 
	the partition sum, phase transition etc. in this modified versions using the same 
	techniques from statistical physics as for the Ising model (transfer matrix for small 
	dimension, mean field method for large dimension resp. branching number).
	\item Alternatively, one might introduce a relaxation time, so the model gains a time 
	dependence. This could be interesting to study non-stationary behaviour.
	\end{itemize}
\end{question}
\begin{question}
	Can there be obtained statistical real-world evidence (and quantified), that the 
	existence of inter-dependency on the path between deciding players and actual decision
	(as modelled in this article) increases cooperation?
\end{question}

\end{document}